\documentclass[10pt]{llncs}
\usepackage[T1]{fontenc}
\usepackage[latin9]{inputenc}
\usepackage{color}
\usepackage{float}
\usepackage{amsmath}
\usepackage{amssymb}

\makeatletter

\floatstyle{ruled}
\newfloat{algorithm}{tbp}{loa}
\floatname{algorithm}{Algorithm}

\usepackage{llncsdoc}\@ifundefined{definecolor}
 {\usepackage{color}}{}
\usepackage{float}\usepackage{amsfonts}\usepackage{multicol}
\usepackage{fancyhdr}

\newtheorem{rem}{Remark}

\makeatother
\fancypagestyle{plain}{%
\fancyhf{} 
\fancyfoot[L]{Preprint to appear in COCOON 2011} 

}

\begin{document}

\title{Unconstrained and Constrained Fault-Tolerant Resource Allocation%
\thanks{This work was partially supported by Australian Research Council Discovery
Project grant \#DP0985063.%
}}

\author{Kewen Liao \and Hong Shen\\
}

\institute{\institute{School of Computer Science \\ The University of Adelaide, SA 5005, Australia \\ \email{\{kewen, hong\}@cs.adelaide.edu.au}}}
\maketitle
\begin{abstract}
First, we study the Unconstrained Fault-Tolerant Resource Allocation
(UFTRA) problem (a.k.a. FTFA problem in \cite{shihongftfa}). In the
problem, we are given a set of sites equipped with an unconstrained
number of facilities as resources, and a set of clients with set $\mathcal{R}$
as corresponding connection requirements, where every facility belonging
to the same site has an identical opening (operating) cost and every
client-facility pair has a connection cost. The objective is to allocate
facilities from sites to satisfy $\mathcal{R}$ at a minimum total
cost. Next, we introduce the Constrained Fault-Tolerant Resource Allocation
(CFTRA) problem. It differs from UFTRA in that the number of resources
available at each site $i$ is limited by $R_{i}$. Both problems
are practical extensions of the classical Fault-Tolerant Facility
Location (FTFL) problem \cite{Jain00FTFL}. For instance, their solutions
provide optimal resource allocation (w.r.t. enterprises) and leasing
(w.r.t. clients) strategies for the contemporary cloud platforms. 

In this paper, we consider the metric version of the problems. For
UFTRA with uniform $\mathcal{R}$, we present a star-greedy algorithm.
The algorithm achieves the approximation ratio of 1.5186 after combining
with the cost scaling and greedy augmentation techniques similar to
\cite{Charikar051.7281.853,Mahdian021.52}, which significantly improves
the result of \cite{shihongftfa} using a phase-greedy algorithm.
We also study the capacitated extension of UFTRA and give a factor
of 2.89. For CFTRA with uniform $\mathcal{R}$, we slightly modify
the algorithm to achieve 1.5186-approximation. For a more general
version of CFTRA, we show that it is reducible to FTFL using linear
programming. 
\end{abstract}
\markboth{Kewen Liao and Hong Shen}{Unconstrained and Constrained
Fault-Tolerant Resource Allocation}\thispagestyle{plain} \lfoot{footer text}

\section{Introduction}

Many facility location models are built around the classical Uncapacitated
Facility Location (UFL) problem. Since late 1990s, UFL has been studied
extensively from the perspective of approximation algorithms for its
practical use in clustering applications and network optimization
\cite{jain01approximation}. With the practical focus on the later,
the Fault-Tolerant Facility Location (FTFL) model was proposed in
\cite{Jain00FTFL}, which adds tolerance to UFL against connection/communication
failures between facilities and clients. As our computer world evolves,
optimal resource allocation starts to become essential in many contemporary
applications \cite{chang2010optimal}. The typical example is today's
cloud computing platform. In particular, cloud providers own data
centers equipped with physical machines as resources. Each physical
machine runs a few virtual machines to serve remote clients' leases
for computing power. This new resource sharing paradigm inevitably
brings up two optimization problems from different perspectives: 1)
how do \textit{cloud providers} cost-effectively share their resources
among a larger number of closer clients? 2) how can \textit{clients}
optimally rent machines in clouds that incurs minimum cost? Through
a careful look at these questions, it is interesting to see they may
be interpreted as facility location problems where facility costs
w.r.t. providers are physical machine costs, whereas to clients, they
are administration overheads; Connection costs w.r.t. providers are
distances and they are renting rates to clients. However, existing
models like UFL and FTFL are insufficient to capture the resource
allocation scenario. They both restrict at most one facility/resource
to open at each site/data center, and connection requests/leases of
a client must be served by facilities from different sites. This motivates
us to study the Unconstrained Fault-Tolerant Resource Allocation (UFTRA)
and the Constrained Fault-Tolerant Resource Allocation (CFTRA) models
generalized from FTFL. Both of them allow multiple facilities to open
at every site and a client's connection requests to be served by facilities
within the same site. This is more realistic since a client may access
multiple machines on the same site in parallel. The difference is
the former model does not limit the number of resources to use at
each site (so provider can always add them) while the later constrains
the amount of available resources to allocate. 

\textbf{Related Work: }For UFL, due to its combinatorial structure,
best approximation algorithms are based on primal-dual and LP-rounding
techniques. In the stream of primal-dual algorithms, JV \cite{jain01approximation},
MMS \cite{Mohammad1.861} and JMS \cite{Jain02greedy} algorithms
are significant and well known. They achieved approximation ratios
of 3, 1.861 and 1.61 respectively. Charikar and Guha \cite{Charikar051.7281.853}
improved the result of JV algorithm to 1.853 and Mahdian et al. \cite{Mohammad06FLP}
improved JMS algorithm to 1.52-approximation. Both of these improvements
were made using the cost scaling and greedy augmentation techniques.
The first LP-rounding algorithm for UFL is proposed by Shmoys et al.
\cite{Shmoys97FL}. It achieved a ratio of 3.16 based on the filtering
and rounding technique of Lin and Vitter \cite{Lin92filting}. Over
the past decade, this result has been improved progressively until
very recently Charikar and Li \cite{charikar2011novel} gave the current
best ratio of 1.488. In contrast to UFL, primal-dual methods suffered
from the bottleneck \cite{Jain00FTFL} for the general non-uniform
FTFL. Constant results were only for the uniform case. In particular,
Jain et al. \cite{Jain03dualfitting} showed that their MMS and JMS
algorithms can be adapted to FTFL with the same approximation factors
of 1.861 and 1.61. Swamy and Shmoys \cite{Swamy08FTFL2.076} improved
this to 1.52. On the other hand, based on LP-rounding Guha et al.
\cite{Guha01FTFL2.47,Guha03FTFL2.41} obtained the first constant
factor algorithm for the non-uniform FTFL with ratio 2.408. Later,
it was improved to 2.076 by Swamy and Shmoys \cite{Swamy08FTFL2.076}.
Recently, Byrka et al. \cite{JaroslawFTFL1.725} achieved the current
best ratio of 1.7245 using dependent rounding. 

Guha and Khuller \cite{Guha99greedy} proved a lower bound of 1.463
for UFL holds unless $P=NP$ \cite{chudak2005improved}. The ratio
also bounds FTFL and the problems we study here. UFTRA is also known
as the fault-tolerant facility allocation (FTFA) problem in the context
of \cite{shihongftfa}. Xu and Shen used a phase-greedy algorithm
to obtain approximation ratio of 1.861. Recently, Yan and Chrobak
\cite{yan2011approximation} gave a rounding algorithm that achieved
3.16-approximation.

\textbf{Our Results:} We present a star-greedy algorithm for UFTRA
with uniform $\mathcal{R}$ that first achieves approximation factor
of 1.61. Our algorithm is motivated by JMS \cite{Jain03dualfitting}
and its adaptation for FTFL \cite{Swamy08FTFL2.076}. We also give
an equivalent primal-dual algorithm, and apply the dual fitting \cite{Jain03dualfitting}
and the inverse dual fitting \cite{shihongftfa} techniques for its
ratio analysis. Together with the cost scaling and greedy augmentation
techniques similar to \cite{Charikar051.7281.853,Mahdian021.52},
the overall algorithm arrives at the current best ratio of 1.5186
which significantly improves the 1.861-approximation in \cite{shihongftfa}.
In addition, we study the capacitated version of UFTRA and provide
a factor of 2.89. For CFTRA, we show that with a slight modification,
our algorithm for UFTRA preserves the ratio of 1.5186 for its uniform
case. Finally, we use linear programming to formally prove that CFTRA
even with arbitrary facility costs existing on the same site is pseudo-polynomial
time reducible to FTFL.

\section{Unconstrained FTRA }

In UFTRA,\textcolor{black}{{} we are given a set of sites $\mathcal{F}$
and a set of clients $\mathcal{C}$, where $\left|\mathcal{F}\right|=n_{f}$
and $\left|\mathcal{C}\right|=n_{c}$. For convenience, let $n=n_{f}+n_{c}$.
At each site $i\in\mathcal{F}$, an unbounded number of facilities
with $f_{i}$ as costs can be opened. There is also a connection cost
$c_{ij}$ between each client $j\in\mathcal{C}$ and all facilities
at site $i$. The objective is to optimally allocate a certain number
facilities from each $i$ to serve every client $j$ with $r_{j}\in\mathcal{R}$
requests while minimizing the total cost incurred. UFTRA can be formulated
by the following integer linear program (ILP). }In the formulation,
\textcolor{black}{variable $y_{i}$ denotes in the solution the number
of facilities to open at site $i$, and $x_{ij}$ the number of connections
between the ($i$, $j$) pair. Compared to the FTFL problem} \cite{Jain00FTFL},
domains of these variables are relaxed to be non-negative rather than
0-1 integers and therefore UFTRA forms a relaxation of FTFL.

{\small \begin{equation}
\begin{array}{llc}
\mathrm{minimize} & \sum_{i\in\mathcal{F}}f_{i}y_{i}+\sum_{i\in\mathcal{F}}\sum_{j\in\mathcal{C}}c_{ij}x_{ij}\\
\mathrm{subject\, to} & \forall j\in\mathcal{C}:\,\sum_{i\in\mathcal{F}}x_{ij}\ge r_{j}\\
 & \forall i\in\mathcal{F},j\in\mathcal{C}:\, y_{i}-x_{ij}\geq0\\
 & \forall i\in\mathcal{F},j\in\mathcal{C}:\, x_{ij}\in\mathbb{Z}^{+}\\
 & \forall i\in\mathcal{F}:\, y_{i}\in\mathbb{Z}^{+}\end{array}\label{eq:ucftra-ilp}\end{equation}
}{\small \par}

\subsection{The Algorithms}

Xu and Shen's algorithm \cite{shihongftfa} for \textcolor{black}{UFTRA}
runs in phases. In each phase, clients that have not fulfilled requirements
get connected to one more facility. In contrast to their phase-greedy
approach, our algorithm iteratively picks the star with the least
average cost (the most cost-effectiveness) and at the same time optimizes
the overall connection cost. It terminates until all clients' connection
requirements are satisfied. In Algorithm 1, we incrementally build
the solution $y_{i}$'s and $x_{ij}$'s which are initially set to
$0$. Their values will then increase according to the star being
picked. We define set $\mathcal{U}$ includes all clients that have
not fulfilled their connection requirements. In order to ensure the
feasibility of the solution, two conditions need to be met while iteratively
choosing stars: 1) the previously opened and used facility at a site
will have zero opening cost in the next iteration; 2) validity of
stars, i.e., a star to be chosen only consists of a facility and clients
have not connected to it. For these conditions, we consider \textit{two
types} of facilities for every site: closed facilities with cost $f_{i}$
and already opened ones with no cost. Also, w.r.t. the closed facilities
of site $i$, we construct the set of clients to be chosen by $C1_{i}=\mathcal{\mathcal{U}}$.
Similarly for the previously opened facilities of site $i$, the target
clients are put into the set $C2_{i}=\left\{ j\in\mathcal{\mathcal{U}}\,|\, x_{ij}<y_{i}\right\} $.
Initially $\forall i\in\mathcal{F}:\, C1_{i}=\mathcal{C}$, $C2_{i}=\emptyset$.
Therefore, at every $i$ the star to be selected is either $\left\{ \left(i,\, C'\right)\,|\, i\in\mathcal{F},\, C'\subseteq C1_{i}\right\} $
with facility cost $f_{i}$ or $\left\{ \left(i,\, C'\right)\,|\, i\in\mathcal{F},\, C'\subseteq C2_{i}\right\} $
with facility cost $0$. In addition, since each client $j$ has at
least $r_{j}$ demands, we can treat them as virtual ports in $j$
with each of them to be assigned to a single facility. W.l.o.g., we
number the ports of $j$ from $1$ to $r_{j}$ and connect them in
ascending order. In every iteration of the algorithm, we use variable
$p_{j}$ to keep track of the port of client $j$ to be connected.
Initially\textcolor{black}{{} $\forall j\in\mathcal{C}:\, p_{j}=1$},
and obviously $\mathcal{U}=\left\{ j\in\mathcal{C}\,|\, p_{j}\leq r_{j}\right\} $.
Moreover, the optimization of the overall connection cost actually
happens when a closed facility is opened and some clients in $\mathcal{C}\backslash\mathcal{\mathcal{U}}$
switch their most expensive connections to the facility. In order
to capture this, we denote the port $q$ of $j$ as $j^{\left(q\right)}$
where $1\leq q\leq r_{j}$, and $\phi\left(j^{\left(q\right)}\right)$
as the site $j^{\left(q\right)}$ connects to. Therefore, the\textit{
combined greedy objective} for picking the most cost-effective star
in Algorithm 1 is defined as the minimum of {\footnotesize $\min_{i\in\mathcal{F},\, C'\subseteq C2_{i}}\frac{\sum_{j\in C'}c_{ij}}{\left|C'\right|}$}
and {\footnotesize $\min_{i\in\mathcal{F},\, C'\subseteq C1_{i}}\frac{f_{i}+\sum_{j\in C'}c_{ij}-\sum_{j\in\mathcal{C}\backslash\mathcal{\mathcal{U}}}\max\left(0,\,\max_{q}c_{\phi\left(j^{\left(q\right)}\right)j}-c_{ij}\right)}{\left|C'\right|}$}.

We restate Algorithm 1 as an equivalent primal-dual algorithm (Algorithm
2) for the sake of ratio analysis. In addition to the previous definitions,
each port $j^{\left(q\right)}$ of client $j$ is associated with
a dual variable $\alpha_{j}^{q}$, representing the total price paid
by port $j^{\left(q\right)}$. We also denote a time\textcolor{black}{{}
$t$}, \textcolor{black}{which increases monotonically} from 0.\textcolor{black}{{}
At any $t$, }we define the contribution of $j$ to site $i$ as\textcolor{black}{{}
\eqref{eq:contribution} and the event $j$ connects to $i$ happens
in two cases: 1) $j$ fully pays the connection cost of an open facility
at $i$ that it is not connected to; 2) total contribution to a closed
facility at $i$ fully pays its opening cost and $j$'s contribution
is positive.}{\footnotesize \begin{equation}
\begin{cases}
\max\left(0,\, t-c_{ij}\right) & \textrm{if}\, j\in\mathcal{\mathcal{U}}\\
\max\left(0,\,\max_{q}c_{\phi\left(j^{\left(q\right)}\right)j}-c_{ij}\right) & \textrm{if}\, j\in\mathcal{\mathcal{\mathcal{C}\backslash U}}\end{cases}\label{eq:contribution}\end{equation}
}{\footnotesize \par}
\begin{lemma}
Runtime complexity of the Primal-Dual Algorithm is $O\left(n^{3}\max_{j}r_{j}\right)$.\end{lemma}
\begin{proof}
Clients' reconnections dominate the time complexity. Once they happen
in Event 2, it takes time $O\left(n_{c}n_{f}\right)$ to update clients'
contributions to other facilities for computing anticipated time of
events. There are maximum $\sum_{j}r_{j}$ such events, therefore
total time is $O\left(\sum_{j}r_{j}n_{f}n_{c}\right)$, i.e. $O\left(n^{3}\max_{j}r_{j}\right)$.
\end{proof}
\begin{algorithm}[H]
{\footnotesize \caption{Star-Greedy Algorithm}
}{\footnotesize \par}

\textbf{Input}: $\forall i,\, j:\, f_{i}$, $c_{ij}$, $r_{j}$.

\textbf{Output}: $\forall i,\, j:\, y_{i}$, $x_{ij}$.

\textbf{Initialization:} Set $\mathcal{U}=\mathcal{C}$, $\forall i,\, j:\, y_{i}=0,\, x_{ij}=0,\, p_{j}=1$.

\medskip{}

While $\mathcal{U}\neq\mathcal{\emptyset}$:
\begin{enumerate}
\item Choose the optimal star $\left(i,\, C'\right)$ according to the combined
greedy objective.
\item If $\exists j\in C'$: $x_{ij}=y_{i}$, then set $y_{i}=y_{i}+1$.
\item $\forall j\in C':$ set $\phi\left(j^{\left(p_{j}\right)}\right)=i$
and $x_{ij}=x_{ij}+1$; $\forall j\in\mathcal{C}\backslash\mathcal{U}$
s.t. $\max_{q}c_{\phi\left(j^{\left(q\right)}\right)j}-c_{ij}>0:$
set $x_{\phi\left(j^{\left(q\right)}\right)j}=x_{\phi\left(j^{\left(q\right)}\right)j}-1$,
$x_{ij}=x_{ij}+1$ and $\phi\left(j^{\left(q\right)}\right)=i$.
\item $\forall j\in C'$ s.t. $p_{j}=r_{j}$: set \textbf{$\mathcal{U}=\mathcal{U}\backslash\left\{ j\right\} $},
otherwise set $p_{j}=p_{j}+1$.
\end{enumerate}
\end{algorithm}

\begin{algorithm}[H]
{\small \caption{Primal-Dual Algorithm}
}{\small \par}

\textbf{Input}: $\forall i,\, j:\, f_{i}$, $c_{ij}$, $r_{j}$.

\textbf{Output}: $\forall i,\, j:\, y_{i}$, $x_{ij}$.

\textbf{\textcolor{black}{Initialization}}\textcolor{black}{: }Set
$\mathcal{U}=\mathcal{C}$, $\forall i,\, j:\, y_{i}=0,\, x_{ij}=0,\, p_{j}=1$.

\medskip{}

\textcolor{black}{While $\mathcal{U}\neq\emptyset$, increase time
$t$ uniformly and execute the events below:} 
\begin{itemize}
\item \textcolor{black}{Event 1: $\exists i\in\mathcal{F},\, j\in\mathcal{U}$:
$t=c_{ij}$ and $x_{ij}<y_{i}$ }\\
 \textcolor{black}{Action 1: Set }$\phi\left(j^{\left(p_{j}\right)}\right)=i$,
$x_{ij}=x_{ij}+1$ and $\alpha_{j}^{p_{j}}=t$; If $p_{j}=r_{j}$,
then set\textbf{ $\mathcal{U}=\mathcal{U}\backslash\left\{ j\right\} $},
otherwise set $p_{j}=p_{j}+1$.\textcolor{black}{{} }\smallskip{}

\item \textcolor{black}{Event 2: $\exists i\in\mathcal{F}$: $\sum_{j\in\mathcal{U}}\max\left(0,\, t-c_{ij}\right)+$$\sum_{j\in\mathcal{C}\backslash\mathcal{U}}\max\left(0,\,\max_{q}c_{\phi\left(j^{\left(q\right)}\right)j}-c_{ij}\right)=f_{i}$}\\
 \textcolor{black}{Action 2: Set $y_{i}=y_{i}+1$}; $\forall j\in\mathcal{C}\backslash\mathcal{U}$
s.t. $\max_{q}c_{\phi\left(j^{\left(q\right)}\right)j}-c_{ij}>0:$
set $x_{\phi\left(j^{\left(q\right)}\right)j}=x_{\phi\left(j^{\left(q\right)}\right)j}-1$,
$x_{ij}=x_{ij}+1$ and $\phi\left(j^{\left(q\right)}\right)=i$; $\forall j\in\mathcal{\mathcal{U}}$
s.t. $t\geq c_{ij}:$ do Action 1.
\end{itemize}
\begin{rem} \textcolor{black}{If more than one event happen at time
$t$, the algorithm processes all of them in an arbitrary order.}
Also, the events themselves may repeatedly happen at any $t$ since
unbounded facilities are allowed to open.\end{rem}%
\end{algorithm}

\subsection{Analysis: Dual Fitting and Inverse Dual Fitting }

Before proceeding our analysis, for simplicity we consider to decompose
any solutions of ILP \eqref{eq:ucftra-ilp} into a collection of stars\textcolor{black}{{}
from set $\mathcal{S}=\left\{ \left(i,\,\mathcal{C}'\right)\,|\, i\in\mathcal{F},\,\mathcal{C}'\subseteq\mathcal{C}\right\} $}
and construct the equivalent ILP \eqref{eq:ucftra-star}. Note that
the star considered here consists of a site and a set of clients.
It is different from the definition in the greedy algorithm where
a star includes two types of facilities. However, this will not make
any difference because $C1_{i}$ and $C2_{i}$ can eventually combine
into a star belonging to $\mathcal{S}$. Moreover, we are allowed
to have duplicate stars in a solution. This directly implies multiple
identical facilities can be opened at every site. The variable $x_{s}$
in \eqref{eq:ucftra-star} denotes the number of duplicate star $s$.
Also, the cost of $s$ denoted by $c_{s}$ is equal to $f_{s}+\sum_{j\in s\cap\mathcal{C}}c_{sj}$.
Here we use $s$ to index the site in star $s$, therefore $f_{s}$
is the facility cost of site $s$ and $c_{sj}$ is the connection
cost between the site and client $j$.

\begin{eqnarray}
 & \mathrm{minimize} & {\displaystyle \sum_{s\in\mathcal{S}}c_{s}x_{s}}\nonumber \\
 & \mathrm{subject\, to} & \forall j\in\mathcal{C}:\,{\displaystyle \sum_{s:j\in s}x_{s}\geq r_{j}}\label{eq:ucftra-star}\\
 &  & \forall s\in\mathcal{S}:\, x_{s}\in\mathbb{Z}^{+}\nonumber \end{eqnarray}

Its LP-relaxation and dual LP are the following:

\setlength{\columnsep}{2pt}
\begin{multicols}{2}

\begin{eqnarray}
 & \mathrm{minimize} & {\displaystyle \sum_{s\in\mathcal{S}}c_{s}x_{s}}\nonumber \\
 & \mathrm{subject\, to} & \forall j\in\mathcal{C}:\,{\displaystyle \sum_{s:j\in s}x_{s}\geq r_{j}}\label{eq:ucftra-star-relax}\\
 &  & \forall s\in\mathcal{S}:\, x_{s}\geq0\nonumber \end{eqnarray}

\columnbreak 

\begin{eqnarray}
 & \textrm{maximize} & {\displaystyle \sum_{j\in\mathcal{C}}r_{j}\alpha_{j}}\nonumber \\
 & \mathrm{subject\, to} & \forall s\in\mathcal{S}:\,{\displaystyle \sum_{j\in s\cap\mathcal{C}}\alpha_{j}\leq c_{s}}\label{eq:ucftra-star-dual}\\
 &  & \forall j\in\mathcal{C}:\,\alpha_{j}\geq0\nonumber \end{eqnarray}

\end{multicols}

\textbf{Single Factor Analysis: }We apply the dual fitting technique
\cite{Jain03dualfitting} for the primal-dual algorithm's single factor
analysis. In order to utilize the weak duality relationship between
LP \eqref{eq:ucftra-star-dual} and LP \eqref{eq:ucftra-star-relax},
we need an algorithm that produces feasible primal ($x_{s}$'s) and
dual ($\alpha_{j}$'s) solutions. Denote the objective values of LPs
\eqref{eq:ucftra-star}, \eqref{eq:ucftra-star-relax} and \eqref{eq:ucftra-star-dual}
by $SOL_{ILP},\, SOL_{LP}$ and $SOL_{D}$ respectively, such an algorithm
establishes the relationship $SOL_{D}\leq SOL_{LP}\leq SOL_{ILP}$.
Note that $SOL_{D}\leq SOL_{LP}$ implies \textit{any} feasible $SOL_{D}$
is upper bounded by \textit{all} feasible $SOL_{LP}$, then apparently
after defining the optimal values of \eqref{eq:ucftra-star-relax}
and \eqref{eq:ucftra-star} as $OPT_{LP}$ and $OPT_{ILP}$ respectively,
we have $SOL_{D}\leq OPT_{LP}\leq OPT_{ILP}.$ However, our algorithm
produce a feasible primal solution but infeasible dual. This is because
some stars may overpay $c_{s}$ and therefore violate the constraint
of \eqref{eq:ucftra-star-dual}. Nevertheless, if we shrink the dual
by a factor $\rho$ and prove the fitted dual $\frac{\alpha_{j}}{\rho}$
is feasible, we get $\frac{SOL_{D}}{\rho}\leq SOL_{LP}\leq SOL_{ILP}$.
Therefore, if we denote $SOL_{P}$ as the total cost of the primal
solution produced by our algorithm, the \textit{key }steps to obtain
the approximation factor are: 1) establish a relationship between
$SOL_{P}$ and $SOL_{D}$ from our primal-dual algorithm; 2) find
a minimum $\rho$ and prove the fitted dual $\frac{\alpha_{j}}{\rho}$
is feasible. For step 1), we have the following lemmas:
\begin{lemma}
The total cost of the primal solution $SOL_{P}$ produced by the Primal-Dual
Algorithm is $\sum_{j\in\mathcal{C}}\sum_{1\leq q\leq r_{j}}$$\alpha_{j}^{q}$.
\label{lem:1}\end{lemma}
\begin{proof}
It is clear that in the algorithm, the sum of dual values of all ports
fully pays all facility and connection costs even with reconnection
of clients. Then the lemmas follows.\end{proof}
\begin{lemma}
Let the dual solution $\alpha_{j}^{r_{j}}$ returned by the Primal-Dual
Algorithm be a solution to LP \eqref{eq:ucftra-star-dual}, i.e. $\alpha_{j}=\alpha_{j}^{r_{j}}$,
then the corresponding $SOL_{D}\geq SOL_{P}$.\label{lem:p=00003Dd}\end{lemma}
\begin{proof}
For a city $j$, $\alpha_{j}^{r_{j}}$ is the largest dual among its
ports. Because we let $\alpha_{j}=\alpha_{j}^{r_{j}}$ in LP \eqref{eq:ucftra-star-dual},
$SOL_{D}=\sum_{j\in\mathcal{C}}r_{j}\alpha_{j}^{r_{j}}\geq\sum_{j\in\mathcal{C}}\sum_{1\leq q\leq r_{j}}\alpha_{j}^{q}=SOL_{P}$.
\end{proof}
For step 2), if we find a minimum $\rho$ s.t. the fitted dual $\frac{\alpha_{j}}{\rho}$
($\alpha_{j}=\alpha_{j}^{r_{j}}$, from now on we use $\alpha_{j}$
for simplicity) is feasible, we will then get $SOL_{D}\leq\rho\cdot OPT_{ILP}$.
Together with the previous lemma, our algorithm is $\rho$-approximation.
The following lemma and corollary are immediate.
\begin{lemma}
Fitted dual $\frac{\alpha_{j}}{\rho}$ is feasible iff $\forall s\in\mathcal{S}:\,\sum_{j\in s\cap\mathcal{C}}\alpha_{j}{\displaystyle \leq\rho\cdot c_{s}}$.\end{lemma}
\begin{corollary}
W.l.o.g., assume a star $s$ consists of a site with opening cost
$f_{s}$ and $k$ clients s.t. $\alpha_{1}\leq\alpha_{2}\leq\cdots\leq\alpha_{k}$.
Denote the connection cost of $j$ to the site as $c_{sj}$, then
the fitted dual is feasible iff $\forall s\in\mathcal{S}:\,\sum_{j=1}^{k}\alpha_{j}\leq\rho\cdot\left(f_{s}+\sum_{j=1}^{k}c_{sj}\right)$,
i.e. $\rho\geq\frac{\sum_{j=1}^{k}\alpha_{j}}{f_{s}+\sum_{j=1}^{k}c_{sj}}$.\label{cor:1}
\end{corollary}
In order to find such a $\rho$, we first prove a couple of properties
that our algorithm holds and then use these properties to guide the
construction of a series of factor-revealing programs. Note that although
the following lemmas are analogous to the ones in \cite{Jain03dualfitting,Swamy08FTFL2.076}
for UFL and FTFL, they essentially reveal UFTRA's unique combinatorial
structure which holds properties both from UFL and FTFL.
\begin{lemma}
At time $t=\alpha_{j}-\epsilon,\,$a moment before port $j^{\left(r_{j}\right)}$
first time gets connected (because $\alpha_{j}=\alpha_{j}^{r_{j}}$),
$\forall1\leq h<j<k,$ let $r_{h,j}=\max_{i}c_{ih}$ if port $h^{\left(r_{h}\right)}$
is already connected to a facility of a site, otherwise let $r_{h,j}=\alpha_{h}$
$\left(\alpha_{h}=\alpha_{j}\right)$, then $\, r_{h,j}\geq r_{h,j+1}$.\label{lem:r}\end{lemma}
\begin{proof}
A client's ports always reconnect to a facility of a site with less
connection cost, so its maximum connection cost will never increase.
The lemma follows.\end{proof}
\begin{lemma}
For any star $s$ with $k$ clients, $\forall1\leq j\leq k:\,\sum_{h=1}^{j-1}\max\left(r_{h,j}-c_{sh},\,0\right)+\sum_{h=j}^{k}\max$$\left(\alpha_{j}-c_{sh},\,0\right)$$\leq f_{s}$.\label{lem:contri}\end{lemma}
\begin{proof}
The lemma follows because at time $t=\alpha_{j}-\epsilon$, in the
primal-dual algorithm the contribution of all clients (either connected
or unconnected) in star $s$ will not exceed the facility's opening
cost at site $s$.\end{proof}
\begin{lemma}
For clients $h,\, j$ in any star $s$ with $k$ clients s.t. $1\leq h<j\leq k:$
$r_{h}=r_{j}=r$, then ${\displaystyle \alpha_{j}\leq r_{h,j}+c_{sh}+c_{sj}}$.\label{lem:tri}\end{lemma}
\begin{proof}
This is where we must enforce all clients have uniform $\mathcal{R}$.
At time $t=\alpha_{j}-\epsilon,\,$if port $h^{\left(r_{h}\right)}$
is still not connected, by Lemma \ref{lem:r} $\alpha_{j}=r_{h,j}$
and this lemma holds. Otherwise, client $h$'s ports have already
connected to $r$ different facilities (not necessary on different
sites) and $r_{h,j}=\max_{i}c_{ih}$. At time $t$, since $j$ has
at most $r-1$ connections, there is at least a facility s.t. $h$
connects to it but $j$ does not. Denote this facility by $i'$, by
triangle inequality we have $c_{i'j}\leq c_{sj}+c_{sh}+c_{i'h}$.
Also because $i'$ is already open, then $\alpha_{j}\leq c_{i'j}$.
The lemma holds from $r_{h,j}=\max_{i}c_{ih}\geq c_{i'h}$.\end{proof}
\begin{theorem}
Let $\rho=sup_{k\geq1}\left\{ \lambda_{k}\right\} $, i.e. the least
upper bound of $\lambda_{k}$ among all $k$ and\label{thm:FRL}{\small \begin{align}
\lambda_{k}=\textrm{maximize\,\,\,\,\,} & {\displaystyle \frac{\sum_{j=1}^{k}\alpha_{j}}{f+\sum_{j=1}^{k}d_{j}}}\nonumber \\
\textrm{subject to\,\,\,\,\,} & \forall1\leq j<k:\alpha_{j}\leq\alpha_{j+1}\nonumber \\
 & \forall1\leq h<j<k:r_{h,j}\geq r_{h,j+1}\nonumber \\
 & \forall1\leq h<j\leq k:{\displaystyle \alpha_{j}\leq r_{h,j}+d_{h}+d_{j}}\label{eq:FRP}\\
 & 1\leq j\leq k:\sum_{h=1}^{j-1}\max\left(r_{h,j}-d_{h},\,0\right)+\sum_{h=j}^{k}\max\left(\alpha_{j}-d_{h},\,0\right)\leq f\nonumber \\
 & 1\leq h\leq j<k:\alpha_{j},\, d_{j},\, f,\, r_{h,j}\geq0\nonumber \end{align}
}{\small \par}

Then the fitted dual is feasible.\end{theorem}
\begin{proof}
Let $f=f_{s},\, d_{j}=c_{sj}$ together with $\alpha_{j},\, r_{h,j}$
constitute a feasible solution to the above program due to Lemma \ref{lem:r},
\ref{lem:contri} and \ref{lem:tri}. Hence $\forall s\in S,\,{\displaystyle \frac{\sum_{j=1}^{k}\alpha_{j}}{f_{s}+\sum_{j=1}^{k}c_{sj}}}$
$\leq\lambda_{k}\leq\rho$ and then theorem follows from Corollary
\ref{cor:1}.\end{proof}
\begin{theorem}
The Primal-Dual Algorithm and Star-Greedy Algorithm achieve 1.61-approximation
for UFTRA with uniform $\mathcal{R}$.\end{theorem}
\begin{proof}
The previous theorem and the weak duality theorem imply when $\rho=sup_{k\geq1}\left\{ \lambda_{k}\right\} $,
$SOL_{D}\leq$$\rho\cdot OPT_{ILP}$. Together with Lemma \ref{lem:p=00003Dd},
it concludes our algorithms are $\rho$-approximation. Also the factor-revealing
program \eqref{eq:FRP} we obtained is equivalent to program (25)
of \cite{Jain03dualfitting}, then we can directly use its result
to get $\forall k,\,\lambda_{k}\leq1.61$ and hence $\rho=1.61$.\medskip{}

\textbf{Bi-Factor Analysis: }We apply the inverse dual fitting technique
\cite{shihongftfa} to the primal-dual algorithm's bi-factor analysis
for its simplicity compared to dual fitting. Inverse dual fitting
considers scaled instances of the problem, and shows the duals of
original instances are feasible to the scaled instances. For UFTRA,
we scale any original instance $\mathcal{I}$'s facility cost by $\rho_{f}$
and connection cost by $\rho_{c}$ to get an instance $\mathcal{I}'$.
In particular in the original problem, let $SOL_{LP}=F_{SOL}+C_{SOL}$,
where $F_{SOL}$ and $C_{SOL}$ represent the total facility cost
and connection cost (they are possibly fractional) of any $SOL_{LP}$
respectively. In the scaled problem, if we define the corresponding
primal and dual costs as $SOL_{LP}'$ and $SOL_{D}'$ (with dual variable
$\alpha_{j}'$), then clearly $SOL_{LP}'=\rho_{f}\cdot F_{SOL}+\rho_{c}\cdot C_{SOL}$,
and if $\alpha_{j}'=\alpha_{j}$ that is feasible to the scaled problem,
by weak duality and Lemma \ref{lem:p=00003Dd} we have $SOL_{P}\leq SOL_{D}=SOL_{D}'\leq SOL_{LP}'$
and the following lemma and corollary.\end{proof}
\begin{lemma}
The Primal-Dual Algorithm is $\left(\rho_{f},\,\rho_{c}\right)$-approximation
iff $\forall s\in\mathcal{S}:\,$$\sum_{j\in s\cap\mathcal{C}}\alpha_{j}\leq\left(\rho_{f}\cdot f_{s}+\rho_{c}\sum_{j\in s\cap\mathcal{C}}c_{sj}\right)$.\end{lemma}
\begin{corollary}
W.l.o.g., assume a star $s$ consists of a site with opening cost
$f_{s}$ and $k$ clients s.t. $\alpha_{1}\leq\alpha_{2}\leq\cdots\leq\alpha_{k}$.
Denote the connection cost of $j$ to the site as $c_{sj}$, then
the Primal-Dual Algorithm is $\left(\rho_{f},\,\rho_{c}\right)$-approximation
iff $\forall s\in\mathcal{S}:\,\sum_{j=1}^{k}\alpha_{j}\leq\left(\rho_{f}\cdot f_{s}+\rho_{c}\sum_{j=1}^{k}c_{sj}\right)$,
i.e. $\rho_{c}\geq\frac{\sum_{j=1}^{k}\alpha_{j}-\rho_{f}\cdot f_{s}}{\sum_{j=1}^{k}c_{sj}}$.\label{cor:2}
\end{corollary}
Similar to dual fitting, we wish to find the minimum value of $\rho_{c}$
for any $\rho_{f}\geq1$. We can construct a new factor revealing
program with objective function: $\lambda_{k}'=\textrm{maximize\,\,\,\,}{\displaystyle \frac{\sum_{j=1}^{k}\alpha_{j}-\rho_{f}f}{\sum_{j=1}^{k}d_{j}}}$
and having same constraints as the program \eqref{eq:FRP}. Clearly,
if $\rho_{c}=sup_{k\geq1}\left\{ \lambda_{k}'\right\} $, we have
$\forall s\in S,\,{\displaystyle {\displaystyle \frac{\sum_{j=1}^{k}\alpha_{j}-\rho_{f}f_{s}}{\sum_{j=1}^{k}c_{sj}}}\leq\lambda_{k}'\leq\rho_{c}}$,
which implies a $\left(\rho_{f},\,\rho_{c}\right)$-approximation
from Corollary \ref{cor:2}. Further, this program is equivalent to
program (36) of \cite{Jain03dualfitting}. Therefore from the result
of \cite{Mahdian021.52}, the Star-Greedy Algorithm is (1.11, 1.78)-approximation.
Finally, after the scaling of facility costs with factor 1.504 and
the similar greedy augmentation that runs in time $O\left(n^{3}\max_{j}r_{j}\right)$
by considering total $n_{f}\max_{j}r_{j}$ facilities, it is easy
to see the overall algorithm achieves the ratio of 1.5186. Details
are omitted.
\begin{theorem}
Star-Greedy Algorithm with cost scaling and greedy augmentation is
1.5186-approximation in time $O\left(n^{3}\max_{j}r_{j}\right)$.
\end{theorem}

\subsection{Capacitated UFTRA }

We observe that there is a strong connection between the well studied
Soft Capacitated Facility Location (SCFL) problem \cite{Shmoys97FL,jain01approximation,Mahdian021.52}
and the Capacitated UFTRA (CUFTRA) problem we consider here. In SCFL,
a facility $i$ is allowed to open multiple times with identical cost
$f_{i}$. This is similar to CUFTRA where a site has unconstrained
resources to allocate. We formulate the CUFTRA problem as ILP \eqref{eq:cuftra},
in which the third constraint limits the total requests a site is
able to serve (capacity of the site). Through investigating the work
for SCFL in \cite{Mahdian021.52}, we discover that the similar result
also holds for CUFTRA.

{\small \begin{equation}
\begin{array}{llc}
\mathrm{minimize} & \sum_{i\in\mathcal{F}}f_{i}y_{i}+\sum_{i\in\mathcal{F}}\sum_{j\in\mathcal{C}}c_{ij}x_{ij}\\
\mathrm{subject\, to} & \forall j\in\mathcal{C}:\,\sum_{i\in\mathcal{F}}x_{ij}\ge r_{j}\\
 & \forall i\in\mathcal{F},j\in\mathcal{C}:\, y_{i}-x_{ij}\geq0\\
 & \forall i\in\mathcal{F}:\,\sum_{j\in\mathcal{C}}x_{ij}\leq u_{i}y_{i}\\
 & \forall i\in\mathcal{F},j\in\mathcal{C}:\, x_{ij}\in\mathbb{Z}^{+}\\
 & \forall i\in\mathcal{F}:\, y_{i}\in\mathbb{Z}^{+}\end{array}\label{eq:cuftra}\end{equation}
}{\small \par}
\begin{theorem}
Any $\left(\rho_{f},\,\rho_{c}\right)$-approximation algorithm for
UFTRA implies a $\left(\rho_{f}+\rho_{c}\right)$-algorithm for CUFTRA.\end{theorem}
\begin{proof}
With the generalized Lagrangian relaxation technique similar to \cite{Mahdian021.52},
we can move the third constraint of CUFTRA into its objective function,
thereby constructing a new UFTRA problem as a relaxation of CUFTRA.
Afterwards, we scale the UFTRA instance's facility costs by $\frac{\rho_{c}}{\rho_{f}}$
and solve the instance using the Star-Greedy Algorithm with output
$Y_{i}$'s and $X_{ij}$'s. Finally we can prove $y_{i}=\max\left(\lceil\sum_{j\in\mathcal{C}}X_{ij}/u_{i}\rceil,\, Y_{i}\right)$%
\footnote{A small bug is fixed here for the COCOON 2011's camera-ready version.%
} and $x_{ij}=X_{ij}$ construct a feasible solution to ILP \eqref{eq:cuftra}
and lead to $\left(\rho_{f}+\rho_{c}\right)$-approximation. Details
are omitted.
\end{proof}
The following theorem is then immediate from the bi-factor result
of $\left(1.11,\,1.78\right)$ for UFTRA with uniform $\mathcal{R}$.
\begin{theorem}
CUFTRA with uniform $\mathcal{R}$ achieves an approximation factor
of 2.89 in time $O\left(n^{3}\max_{j}r_{j}\right)$.
\end{theorem}

\section{Constrained FTRA }

In CFTRA, the only difference compared to UFTRA is the number of resources
to allocate at site $i$ is limited by $R_{i}$ $\left(R_{i}\geq1\right)$.
However, this constraint introduces a harder problem to solve since
FTFL is a special case of CFTRA when $\forall i:$ $R_{i}=1$. Also,
from the practical point of view, the CFTRA model plays an important
role in the resource constrained allocation. The problem's LP-relaxation
and dual are displayed below.

\setlength{\columnsep}{10pt}
\begin{multicols}{2}

{\small \[
\begin{array}{llc}
\mathrm{minimize} & \sum_{i\in\mathcal{F}}f_{i}y_{i}+\sum_{i\in\mathcal{F}}\sum_{j\in\mathcal{C}}c_{ij}x_{ij}\\
\mathrm{subject\, to} & \forall j\in\mathcal{C}:\,\sum_{i\in\mathcal{F}}x_{ij}\ge r_{j}\\
 & \forall i\in\mathcal{F},j\in\mathcal{C}:\, y_{i}-x_{ij}\geq0\\
 & \forall i\in\mathcal{F}:\, y_{i}\leq R_{i}\\
 & \forall i\in\mathcal{F},j\in\mathcal{C}:\, x_{ij}\geq0\\
 & \forall i\in\mathcal{F}:\, y_{i}\geq0\end{array}\]
}{\small \par}

{\small \columnbreak }{\small \par}

{\small \[
\begin{array}{llc}
\textrm{maximize} & \sum_{j\in\mathcal{C}}r_{j}\alpha_{j}-\sum_{i\in\mathcal{F}}z_{i}\\
\mathrm{subject\, to} & \forall i\in\mathcal{F}:\,\sum_{j\in\mathcal{C}}\beta_{ij}\leq f_{i}+\frac{z_{i}}{R_{i}}\\
 & \forall i\in\mathcal{F},j\in\mathcal{C}:\,\alpha_{j}-\beta_{ij}\leq c_{ij}\\
 & \forall i\in\mathcal{F},j\in\mathcal{C}:\,\beta_{ij}\geq0\\
 & \forall j\in\mathcal{C}:\,\alpha_{j}\geq0\\
 & \forall i\in\mathcal{F}:\, z_{i}\geq0\end{array}\]
}{\small \par}

\end{multicols}

After adding an extra constraint $y_{i}<R_{i}$ to Event 2 of the
Primal-dual Algorithm for UFTRA, it is clear that the slightly modified
algorithm computes a feasible primal solution to CFTRA. The question
left is whether the same approximation ratio preserves for CFTRA with
uniform $\mathcal{R}$. The first observation we make is that the
Lemma \ref{lem:contri} for UFTRA fails to hold for CFTRA since the
number of facilities at a site is limited. Therefore, results of UFTRA
do not directly lead to CFTRA's approximation guarantee. In fact,
CFTRA's combinatorial structure generalizes FTFL's. So now we try
to extend the solution to the uniform FTFL in \cite{Swamy08FTFL2.076}.
W.l.o.g., we set {\small $z_{i}=\sum_{j}\theta_{ij}=\begin{cases}
\sum_{j}x_{ij}\left(\alpha_{j}-\alpha_{j}^{l}\right) & \textrm{primarily}\, x_{ij}=R_{i}\\
\sum_{j}0 & \textrm{otherwise}\end{cases}$} where $l$ denotes the last port of $j$ that connects to $i$.
Then we have $SOL_{D}=\sum_{j\in\mathcal{C}}r_{j}\alpha_{j}-\sum_{i\in\mathcal{F}}z_{i}\geq SOL_{P}$
(similar to Lemma \ref{lem:p=00003Dd}). Afterwards, using dual fitting
we can prove $\forall i\in\mathcal{F}:$ $\sum_{j\in\mathcal{C}}\left(\alpha_{j}-\frac{\theta_{ij}}{R_{i}}\right)\leq\rho\left(f_{i}+\sum_{j}c_{ij}\right)$
which implies 1.5186-approximation for CFTRA using inverse dual fitting.
Details are omitted. 
\begin{theorem}
CFTRA with uniform $\mathcal{R}$ can be approximated with a factor
of 1.5186 in time $O\left(n^{3}\max_{j}r_{j}\right)$.
\end{theorem}
Moreover, we study the generalized CFTRA (GCFTRA) problem where facility
costs on each site are allowed to be different. This problem is more
general in the case that costs of resources at each site are not necessary
identical. For GCFTRA, we have a different problem formulation \eqref{eq:gcftra-ds}
that identifies individual facility $f_{i}^{d}$ at each site $i$. 

{\small \begin{equation}
\begin{array}{llc}
\mathrm{minimize} & \sum_{i=1}^{n_{f}}\sum_{d=1}^{R_{i}}f_{i}^{d}y_{i}^{d}+\sum_{i=1}^{n_{f}}\sum_{j=1}^{n_{c}}c_{ij}x_{ij}\\
\mathrm{subject\, to} & \forall1\leq j\leq n_{c}:\,\sum_{i=1}^{n_{f}}x_{ij}\ge r_{j}\\
 & \forall1\leq i\leq n_{f},1\leq j\leq n_{c}:\,\sum_{d=1}^{R_{i}}y_{i}^{d}-x_{ij}\geq0\\
 & \forall1\leq i\leq n_{f},1\leq d\leq R_{i}:\, y_{i}^{d}\leq1\\
 & \forall1\leq i\leq n_{f},1\leq j\leq n_{c}:\, x_{ij}\in\mathbb{Z}^{+}\\
 & \forall1\leq i\leq n_{f},1\leq d\leq R_{i}:\, y_{i}^{d}\in\left\{ 0,1\right\} \end{array}\label{eq:gcftra-ds}\end{equation}
}{\small \par}

Now, we consider to reduce this problem to an FTFL problem \cite{Jain00FTFL}.
Instead of clustering facilities within sites as GCFTRA does, we put
all facilities (totally $\sum_{i}R_{i}$) of a GCFTRA instance together
(without separating them by sites) and consider them as a whole. This
transformation then brings in an FTFL problem shown in ILP \eqref{eq:gcftra-ftfl}.
We prove that GCFTRA is pseudo-polynomial time reducible to FTFL,
i.e. ILPs \eqref{eq:gcftra-ftfl} and \eqref{eq:gcftra-ds} are equivalent.

{\small \begin{equation}
\begin{array}{llc}
\mathrm{minimize} & \sum_{k=1}^{\sum_{i}R_{i}}f_{k}y_{k}+\sum_{k=1}^{\sum_{i}R_{i}}\sum_{j=1}^{n_{c}}c_{kj}x_{kj}\\
\mathrm{subject\, to} & \forall1\leq k\leq\sum_{i}R_{i},1\leq j\leq n_{c}:\, y_{k}-x_{kj}\ge0\\
 & \forall1\leq j\leq n_{c}:\,\sum_{k=1}^{\sum_{i}R_{i}}x_{kj}\geq r_{j}\\
 & \forall1\leq k\leq\sum_{i}R_{i}:\, y_{k}\leq1\\
 & \forall1\leq k\leq\sum_{i}R_{i},1\leq j\leq n_{c}:\, x_{kj}\in\left\{ 0,1\right\} \\
 & \forall1\leq k\leq\sum_{i}R_{i}:\, y_{k}\in\left\{ 0,1\right\} \end{array}\label{eq:gcftra-ftfl}\end{equation}
}{\small \par}
\begin{theorem}
GCFTRA is pseudo-polynomial time reducible to FTFL.\end{theorem}
\begin{proof}
Let $\left(y_{k,}\, x_{kj}\right)$ be any solution of ILP \eqref{eq:gcftra-ftfl},
and in ILP \eqref{eq:gcftra-ds} let $y_{i}^{d}=y_{k}$ if the facility
$i^{\left(d\right)}=k$ and $x_{ij}=\sum_{k\in i}x_{kj}$. Note that
through our transformation, an FTFL instance includes all facilities
of an GCFTRA instance. The condition $i^{\left(d\right)}=k$ denotes
the case of a pair of identical facilities and $k\in i$ denotes the
case if the $k$th facility of an FTFL instance belongs to the $i$th
site of the GCFTRA instance. We first substitute $\left(y_{k,}\,\sum_{k\in i}x_{kj}\right)$
into ILP \eqref{eq:gcftra-ds} and show it constitutes a feasible
solution. Next, it is easy to see the objective values of ILPs \eqref{eq:gcftra-ds}
and \eqref{eq:gcftra-ftfl} are equivalent after substitution. Details
are omitted.
\end{proof}
\bibliographystyle{plain} \bibliographystyle{plain}
\bibliography{FTRAs-Kewen-CoRR}
 
\end{document}